\documentclass[11pt,a4paper]{article}
\usepackage{fullpage,amssymb,amsmath,url}

\usepackage{paralist}
\usepackage{euscript}
\usepackage{xspace}
\usepackage[noend]{algorithmic}
\usepackage{graphicx}
\usepackage{citesort}
\usepackage{url}

\newcommand{\seclab}[1]{\label{sec:#1}}
\newcommand{\secref}[1]{Section~\ref{sec:#1}}

\def\scan(#1){\ensuremath{\mathrm{scan}(#1)}}
\def\sort(#1){\ensuremath{\mathrm{sort}(#1)}}
\def\Scan(#1){\ensuremath{\frac{#1}{B}}}
\def\Sort(#1){\ensuremath{\frac{#1}{B} \log_{M/B} \frac{#1}{B}}}
\def\boundary(#1){\ensuremath{\partial(#1)}}

\newtheorem{lemma}{Lemma}

\newenvironment{proof}{Proof:}{\qed}
\def\squareforqed{\hbox{\rlap{$\sqcap$}$\sqcup$}}
\def\qed{\ifmmode\squareforqed\else{\unskip\nobreak\hfil
\penalty50\hskip1em\null\nobreak\hfil\squareforqed
\parfillskip=0pt\finalhyphendemerits=0\endgraf}\fi}

\title{I/O-optimal algorithms on grid graphs}
\author{Herman~Haverkort\footnote{%
  Dept.\ of Mathematics and Computer Science,
  Eindhoven University of Technology,
  the Netherlands,
  \texttt{cs.herman@haverkort.net}}}

\begin{document}
\maketitle

\begin{abstract}
Given a graph of which the $n$ vertices form a regular two-dimensional grid, and in which each (possibly weighted and/or directed) edge connects a vertex to one of its eight neighbours, the following can be done in $O(\scan(n))$ I/Os, provided $M = \Omega(B^2)$: computation of shortest paths with non-negative edge weights from a single source, breadth-first traversal, computation of a minimum spanning tree, topological sorting, time-forward processing (if the input is a plane graph), and an Euler tour (if the input graph is a tree). The minimum-spanning tree algorithm is cache-oblivious. The best previously published algorithms for these problems need $\Theta(\sort(n))$ I/Os.
Estimates of the actual I/O volume show that the new algorithms may often be very efficient in practice.
\end{abstract}

\section{Introduction}\seclab{Intro}
Many applications work with massive graphs that are too large to fit in the main memory of a computer. Therefore, in computations on such graphs, the bottleneck is often the transfer of data (I/O) between main memory and disk, rather than actual work of the CPU. In making such computations feasible, the primary goal is therefore to make them \emph{I/O-efficient}. Sometimes the graphs have a grid structure: the vertices form a regular two-dimensional grid, and each edge connects a vertex to one of its eight neighbours. An example of this are graphs constructed from elevation models of terrains in geographic information systems, which are used for simulations of hydrological processes, erosion etc. Therefore I/O-efficient computations on such graphs have gotten a fair amount of attention~\cite{GridSSSP,TerraFlow,TerraCost,Pfafstetter,TerraStream}. In this paper we study the computation of single-source shortest paths with non-negative edge weights (SSSP), the computation of minimum spanning trees (MST), topological sorting, and several other problems that can be solved in linear time on grid graphs~\cite{CLRS,LinearSSSP}, but for which, to my knowledge, an I/O-optimal algorithm was not yet known.

We analyse our algorithms in the standard model of Aggarwal and Vitter~\cite{IOModel}. In this model, the machine has a main memory of size $M$ and an external memory (disk) of unbounded size. The CPU can only operate on data that is currently in main memory; it cannot operate directly on data that is currently on disk. The address space of the disk is divided into blocks of $B$ consecutive addresses. It is often assumed that $M = \Omega(B^2)$, and we will do so throughout this paper. Data can only be transferred between main memory and disk by transferring a complete block; such a block transfer is called an I/O. A \emph{cache-aware} algorithm would know $M$ and $B$ and could control exactly which blocks are transferred. A \emph{cache-oblivious} algorithm~\cite{CacheOblivious} would not know $M$ and $B$ and would not take control: it would simply try to access data in the disk's address space directly, and count on the operating system to bring the corresponding block into memory if necessary. To make space for this in memory, the operating system may need to write another block from memory back to disk: the optimal strategy would be to evict the block needed furthest in future and we assume the operating system would do exactly this (Frigo et al.\ argue why this assumption is reasonable~\cite{CacheOblivious}). 

Both in the cache-aware and in the cache-oblivious setting, scanning a contiguous file of $n$ items in this model takes $\scan(n) = \Theta(\Scan(n))$ I/Os, and sorting takes $\sort(n) = \Theta(\Sort(n))$ I/Os~\cite{IOModel,CacheOblivious}. Note that $\scan(n)$ is much less than $n$: it is the difference between doing one I/O roughly every $B$ steps of a computation, or doing one I/O roughly every step of the way. On large files this constitutes a difference in actual running time between several minutes and several years. Since the problems we study all have linear-time algorithms on a random-access machine, we would hope to obtain algorithms that run in $O(\scan(n))$ I/Os. Arge et al.\ already achieved this for labelling connected components~\cite{GridSSSP}.

Spending more than $O(\scan(n))$ I/Os on problems such as breadth-first traversal and topological sorting cannot always be avoided if the vertices of the graph (and their adjacency lists) are not ordered in any way. If there is no structure in the input, any algorithm solving these problems must be able to perform any permutation of the vertices---consider a traversal of a graph that consists of a single path, of which the vertices are given in random order---and permutations require $\Omega(\min(n, \sort(n))) = \omega(\scan(n))$ I/Os in the worst case~\cite{IOModel}. However, if the input is provided in a suitable format, we may be able to do better.

In this paper we will see that several problems on grid graphs (including all problems mentioned above) can be solved by cache-aware algorithms in $O(\scan(n))$ I/Os, provided the vertices and their outgoing edges are given in row-major, column-major or Z-order (see \secref{prelims}). The minimum-spanning tree algorithm is even cache-oblivious. For time-forward processing (see \secref{TFP}), there is a small catch: the input should be a \emph{plane} grid graph.

In addition to an analysis of the asymptotic I/O-complexity, this paper gives calculations of the I/O volume (number of bytes transferred) relative to the sum of the input and the output size of each algorithm. In these calculations we assume we have 2\,GB = $2^{31}$ bytes of main memory; data is transferred in blocks of 128\,KB = $2^{17}$ bytes (reading or writing smaller amounts is possible but counts for 128\,KB of I/O volume); the graph consists of $n = 2^{40}$ vertices; data is stored as 8-, 16-, 32- or 64-bits numbers.
(With these parameters, I/O-efficient merge sort on the input takes three passes and thus has relative I/O volume~3.)
For a justification of the I/O volume computations, see the appendix, \secref{apx:intro}.


At the end of the paper we reflect on the practical potential of our results and
we discuss possible applications, possible extensions to general planar graphs, and remaining questions for further research.

\paragraph{About this manuscript}
\textit{This manuscript consists of a main text that briefly sketches my algorithms for the problems mentioned above, and an appendix that gives more detailed descriptions, proofs, and analysis. This set-up results from the fact that this paper was originally submitted to the European Symposium on Algorithms in 2011. The programme committee of that conference rejected the manuscript, citing reviewers who found the results not interesting enough and found that the paper would be unreadable when published without the appendix. This may be true: the latter results from my attempt to compress the paper into the page limit of the conference. I hope to be able to continue working on this topic at some time in the future, but, since I do not have the time or funding to do so in the near foreseeable future, I decided to put the manuscript on arXiv as it is, with few changes.}

\section{Preliminaries}\seclab{prelims}
We assume that grid graphs are given by listing their vertices one by one, for each vertex giving the weights of the (at most eight) outgoing edges connecting them to their neighbours in the grid. We assume that the vertices are listed in order, either row by row, and within each row from left to right; column by column, and within each column from the top down; or in Z-order. The last order is defined as follows. Suppose a grid graph $G$ has $r$ rows, numbered from $1$ to $r$, and $c$ columns, numbered from $1$ to $c$. Let $k$ be the smallest integer that $2^k \geq r$ and $2^k \geq c$. We divide $G$ into quadrants by cutting it between row $2^k/2$ and $2^k/2+1$, and between column $2^k/2$ and $2^k/2+1$. On disk, we first store the top left quadrant, then the top right quadrant, then the bottom left quadrant, and finally the bottom right quadrant of $G$. Each quadrant is ordered recursively in the same way. Observe that in all three orders, one can compute the row and column numbers of a vertex from its address on disk and vice versa without I/O.

A \emph{canonical cluster} is a quadrant that appears at some level in the recursion that defines the Z-order. A canonical cluster is therefore any block of vertices in row $i \cdot 2^h + 1$ to $(i+1) 2^h$ and column $j \cdot 2^h + 1$ to $(j+1) 2^h$, for some natural numbers $h, i, j$. Clearly, when vertices are stored in Z-order, each canonical cluster occupies a contiguous part of the disk. This will be good for the practical performance of all algorithms in this paper, and necessary for the cache-oblivious algorithm in \secref{COBMST}. Therefore, in the following sections we assume that the input has been converted to Z-order. For grid graphs in row-by-row or in column-by-column order this conversion can be done cache-obliviously in $O(\scan(n))$ I/Os by going through the grid in Z-order~\cite{SimpleFlow}.

We call a canonical cluster $Q$ of width $2^h$ an \emph{$h$-cluster}.
For any set of vertices~$Q$, let $G(Q)$ be the subgraph of $G$ that consists of the vertices~$Q$ and all edges $(u,v)$ of $G$ such that $u, v \in Q$. The \emph{boundary} $\boundary(Q)$ of a canonical cluster~$Q$ is the set of vertices of $Q$ that are adjacent to vertices outside $Q$; for all practical purposes we can assume these are simply all vertices in the topmost and bottommost rows of $Q$ and all vertices in the leftmost and rightmost columns of~$Q$.
We define the \emph{$h$-separator set} $V_h$ of $G$ as the set of vertices that are on the boundaries of the $h$-clusters. The vertices in $V_h$ can be numbered consecutively such that for each $h$-cluster $Q$, the vertices of
$\boundary(Q)$ get consecutive numbers---say in clockwise order around $Q$, starting from the upper left corner. We call these numbers the \emph{$h$-numbers}. Given $h$ and the row and column number of a vertex, its $h$-number, if it exists, can be computed without any I/O.
Let $G^h$~be the graph with vertex set $V_h$ and edge set $E_h$, which consists of the edges $(u,v)$ of $G$ such that $u$ and $v$ are in different $h$-clusters.

\section{Single-source shortest paths (SSSP)}
Given a directed grid graph $G$ with non-negative edge weights, and a source vertex $s$, we want to compute, for each vertex $t$ of the graph, the distance $\delta_G(s,t)$, that is, the (minimum total weight) of a path in $G$ from $s$ to $t$. Arge et al.\ describe a solution that uses $O(\sort(n))$ I/Os~\cite{PlanarSSSP}. We will first improve this to $O(\scan(n) + \sort(n/\sqrt{M}))$ I/Os. This is $O(\scan(n))$ already for all practical purposes, namely, whenever $n \leq (M/B)^{O(\sqrt{M})}$. Nevertheless we will also see an alternative solution of which the asymptotic I/O-complexity is truely $O(\scan(n))$.

\subsection{Single-source shortest paths in $O(\scan(n) + \sort(n/\sqrt{M}))$ I/Os.}\seclab{SimpleSSSP}
The algorithm consists of three phases. First, we construct a graph $G'$ on the vertices of $G^h$, where $h$ is chosen such that $4^h \leq c M$ for a sufficiently small constant $c$. Second, we run Dijkstra's algorithm on $G'$ to compute the distances from $s$ to each vertex of $G^h$. Third, we use these distances to compute the distances from $s$ to the remaining vertices of $G$. The improvement compared to the algorithm from Arge et al.\ lies in a reduction of the size of the priority queue in the second phase (from $O(n/\sqrt{M})$ to $O(n/M)$) and the number of operations performed on it (from $O(n)$ to $O(n/\sqrt{M})$).

In the first phase, we process the $h$-clusters one by one. We read each $h$-cluster $Q$ into memory. We compute outgoing edges for each vertex $u$ of $\boundary(Q)$: these are the edges of $G$ that lead from $u$ to vertices outside $Q$, and for each other vertex $v$ of $\boundary(Q)$, an edge $(u,v)$ with weight $w(u,v) = \delta_{G(Q)}(u,v)$. The outgoing edges for $u$ are added to $G'$. The edges of $G'$ are stored ordered by the $h$-numbers of their tails $u$.

In the second phase, we will maintain a file $D$ that contains the distance estimates $d[t]$ from $s$ to each vertex $t$ of $V_h$, ordered by $h$-number. It also stores, for each vertex, whether the distance estimate is \emph{final} or \emph{tentative}. Initially, we set all estimates to $\infty$, tentative. Then we read the $h$-cluster $Q$ that contains $s$, and set $d[t] = \delta_{G(Q)}(s,t)$, tentative, for each vertex $t \in \boundary(Q)$. We initialize an I/O-efficient priority queue~\cite{PriorityQueue} with one element: a number identifying $Q$, with key $\min_{t \in \boundary(Q)} \delta_{G(Q)}(s,t)$. This completes the initialization of this phase.

The algorithm now proceeds as follows, until the queue is empty. The identifier of the $h$-cluster $Q$ with minimum key is extracted from the priority queue, let its key be $d$. We read the distance estimates for all vertices of $\boundary(Q)$ from $D$. We choose a vertex $u \in \boundary(Q)$ with tentative distance $d$ and make its distance estimate final. Then we read the outgoing edges of $u$ from $G'$, and for each edge $(u,v)$, we check if $d[u] + w(u,v) < d[v]$; if so, we set $d[v] = d[u] + w(u,v)$, tentative. Finally, we update the priority queue: for each of the $h$-clusters $Q'$ that were touched we set its key to the smallest tentative distance of the vertices of $\boundary(Q')$; if all distance estimates for $Q$ are final, it is not put back into the queue.

In the third phase, we process the $h$-clusters one by one as follows. We read each $h$-cluster $Q$ into memory, and run Dijkstra's algorithm on $Q$, initializing the priority queue with the vertices of $\boundary(Q)$ and their distances as stored in~$D$. We write the distances computed for $Q$ to an output file in Z-order.

It remains to choose $c$: this is done such that each $h$-cluster can be processed entirely in main memory in the first and the third phase of the algorithm.

In total, the above algorithm runs in $O(\scan(n) + \sort(n/\sqrt{M}))$ I/Os. The relative I/O volume is less than~13 (see appendix for details).

\subsection{Single-source shortest paths in $O(\scan(n))$ I/Os.}
To get a solution that requires only $O(\scan(n))$ I/Os, we combine the ideas presented above with the approach taken in the linear-time single-source shortest paths algorithm from Henzinger et al.~\cite{LinearSSSP}.

In the second phase of the algorithm, we will work with a hierarchy of canonical clusters, with sizes $4^{h_0} < 4^{h_1} < ... < 4^{h_k}$, where $h_0 = h$ (as defined above), $h_1 = h_0 + 3$, $h_i = 2^{h_{i-1} - h_{i-2} - 2} \cdot h_{i-1}$ for $i > 1$, and $h_k$ is the smallest value in this series such that the full input fits in a single $h_k$-cluster. For example, with $n = 2^{40}$ and $h = 12$, we get $h_1 = 15, h_2 = 30, k = 2$. For each $h_i$-cluster $Q$ with $i > 0$ in the hierarchy there is a priority queue, which stores an identifier of each unfinished $h_{i-1}$-cluster $Q'$ within $Q$; the key of $Q'$ is the lowest tentative distance of any vertex within~$Q'$. After initializing all queues, the algorithm extracts and processes $h_{k-1}$-clusters from the queue of the single $h_k$-cluster, until this queue is empty. An $h_i$-cluster $Q$ with $i > 0$ is processed as follows:

\clearpage
\begin{algorithmic}[1]
\FOR{$j \leftarrow 1$ to $2^{h_i - h_{i-1} - 1}$}
  \IF{the priority queue of $Q$ is not empty}
    \STATE extract an $h_{i-1}$-cluster $Q'$ with minimum key from the queue of $Q$
    \STATE process $Q'$ recursively, obtaining the lowest updated tentative distances for
           $Q'$ and each of its neighbouring $h_{i-1}$-clusters
    \STATE reinsert/update the keys of $Q'$ and its neighbours in the queues of the $h_i$-clusters that contain them
  \ENDIF
\ENDFOR
\STATE return lowest updated tentative distances for $Q$ and neighbouring $h_i$-clusters
\end{algorithmic}

\medskip
The processing of an $h_0$-cluster $Q$ is similar to before: we find a vertex $u$ with a tentative (non-final) distance matching the key of $Q$, we make that distance final, and explore all edges from $u$. In addition, we also return the lowest updated tentative distances within $Q$ and its neighbouring $h_0$-clusters.

The idea is that most work is done within smaller clusters that have small priority queues, while few updates are done on the larger priority queues of the larger clusters. This is achieved as follows: whenever we go into an $h_i$-cluster $Q$, we first do some local work (for $2^{h_i - h_{i-1} - 1}$ iterations) before we consider moving our attention to another $h_i$-cluster. This is risky: while we keep working on~$Q$, we ignore the fact that shorter paths to some vertices of $Q$ may be found by taking a route through another cluster; the work done inside $Q$ may turn out to be premature and the algorithm would redo it later. Henzinger et al.~\cite{LinearSSSP} show how to prove that this does not happen too often in their setting; the same techniques can be applied to the hierarchy as defined by our choice of the sequence $h_0,h_1,...,h_k$. As a result, our algorithm runs in $O(\scan(n))$ I/Os (for details, see the appendix).

\section{Breadth-first traversals}\seclab{BFS}
Given a directed, unweighted grid graph $G$ and a source vertex $s$, we want to put the vertices of $G$ in order such that for any two vertices $u, v$ such that $u$ appears before $v$, we have $\delta_G(s,u) \leq \delta_G(s,v)$. This can be achieved by computing single-source shortest paths and then sorting the vertices by increasing distance from~$s$. In this section we show how to do both steps in $O(\scan(n))$ I/Os.

To compute the shortest path distances, we can use the solutions presented above. Alternatively, we use the first algorithm with a priority queue implementation that supports insertions and deletions in $O(1/B)$ I/Os on average. This is possible because at any time during the algorithm, the keys in the priority queue do not differ more than $\Theta(M)$. Details are found in the appendix.

To sort the vertices in $O(\scan(n))$ I/Os, we first cut the graph into $O(n/\sqrt{M})$ chunks such that within each chunk, the distance of its vertices from $s$ differs by $O(\sqrt{M})$. Then we sort the chunks by their distance from $s$. Finally we read the chunks in order, and write out the vertices in order of distance from $s$. 

To cut the graph into chunks, we process all $h$-clusters one by one. In each $h$-cluster $Q$, we construct a `local' breadth-first search tree, by connecting each non-boundary vertex $v \neq s$ at distance $d$ from $s$ to a neighbour $u$ at distance $d-1$ from $s$. This results in a forest of $\Theta(\sqrt{M})$ trees rooted at the vertices of $\boundary(Q)$ (and $s$, if it lies in $Q$). Any trees that are higher than $\sqrt{M}$ levels are cut into smaller chunks: a tree $T$ of size $t$ is cut into $O(t/\sqrt{M})$ chunks of height at most $\sqrt{M}$, such that each chunk is a connected subgraph of $T$. We write all chunks to a file $C$, describing each chunk by its vertex $u$ which is closest to $s$, the distance $\delta_G(s,u)$ of $u$ from $s$, and the tree rooted at $u$. Furthermore, we write, for each chunk, its address in $C$ and its distance $\delta_G(s,u)$ to a file $A$.

After processing all clusters, we sort $A$ by increasing distance, using an I/O-optimal sorting algorithm or a standard implementation of radix sort with base $\Theta(n/M^{3/4})$---whichever is faster for our specific values of $M$, $B$ and $n$. Finally we put all vertices in order as follows. The idea is that we maintain a set of stacks $S_0,...,S_{n-1}$, such that $S_i$ contains vertices at distance $i$ from $s$; initially all stacks all empty. We now go through the sorted list of chunk addresses $A$. For each chunk we do the following: we retrieve the chunk from~$C$; let $d$ be its distance from $s$. We remove all vertices from the stacks $S_0,...,S_d$ (in order) and output them, and push all vertices from the current chunk on the stacks corresponding to these vertices' distances from $s$. After reading all chunks, we output all vertices that remain on the stacks.

In total, the algorithm uses $O(\scan(n))$ I/Os. The relative I/O volume is less than~100. The high relative I/O volume, compared to the algorithm for shortest paths, is not due to more I/O, but to a more efficient representation of the input.
For analysis and calculations, see the appendix.

\section{Minimum spanning trees (MST)}
Given a connected undirected grid graph $G$ with edge weights, we want to compute a spanning tree of $G$ of minimum total weight. Chiang et al.\ claim a solution in $O(\sort(n))$ I/Os~\cite{PlanarMST}.

Our solutions below are based on the following observation (proof in appendix). Let $G_1,...,G_k$ be the graphs induced by the $h$-clusters $Q_1,...,Q_k$ of $G$, and let $T_1,...,T_k$ be minimum spanning trees of $G_1,...,G_k$, respectively. Then the union $U$ of $T_1,...,T_k$ and $G^h$ contains an MST of $G$. We can use this to compute an MST of $G$ by first computing minimum spanning trees $T_1,...,T_k$ of $G_1,...,G_k$, and then computing an MST of $U$. Two observations will help us speeding up this computation. First, any branches of $T_i$ that do not contain any vertices of $V_h$, must be included in any spanning tree of $U$, because the vertices in such `dead ends' do not have any other connections to the rest of $U$. Second, consider any chain $\gamma = u_0, ..., u_m$ in a tree $T_i$, such that all vertices $u_k (0 < k < m)$ have degree two (after removing dead ends) and are not in $V_h$. In $T$, all vertices of $\gamma$ must be connected along $\gamma$ to either $u_0$ or $u_m$; therefore all edges of $\gamma$, except one, must be included in any spanning tree of $U$. When constructing a spanning tree of $U$, the only decision that we have to make for $\gamma$, is whether to omit one of its heaviest edges. Therefore we can compute an MST $T$ of $G$ as follows.

\subsection{Minimum spanning-trees cache-aware in $O(\scan(n) + \sort(n/\sqrt{M}))$ I/Os.}\seclab{SimpleMST}
For each cluster $G_i$ we compute an MST, we remove all dead ends as described above, and we contract each chain as described above, replacing each chain $u_0,...u_m$ by a single edge $(u_0,u_m)$ with weight $\max_{0 \leq i < m} w(u_i,u_{i+1})$; we call the resulting tree~$T'_i$. Then we compute an MST $T'$ of $U'$, where $U'$ is the union of $T'_1,...,T'_k$ and~$G^h$. Finally, we construct $T$ from $T'$ as follows. We process the graph cluster by cluster. For every chain $(u_0,...,u_m)$ of which the representative edge $(u_0,u_m)$ is \emph{not} included in $T'$, we add all edges except one heaviest edge of the chain to $T$; for every chain $(u_0,...,u_m)$ of which the representative edge $(u_0,u_m)$ \emph{is} included in $T'$, we insert the complete chain in $T$; we also insert all dead ends that had been removed.


Using Prim's algorithm~\cite{CLRS} with an I/O-efficient priority queue~\cite{PriorityQueue} to compute $T'$ from $U'$, the complete algorithm runs in $O(\scan(n) + \sort(n/\sqrt{M}))$ I/Os. When $T'$ can be computed from $U'$ in main memory, the relative I/O volume is approximately 1.5. With 2 GB of memory and 64 bits edge weights, this is the case for files up to 128 GB. Otherwise the algorithm below seems better.

\subsection{Minimum-spanning trees cache-obliviously in $O(\scan(n))$ I/Os.}\seclab{COBMST}
The canonical clusters of the graph form a hierarchy.
We first process this hierarchy in post-order (bottom-up), using two initially empty stacks. The \emph{connections stack} is used to pass spanning trees of clusters and edges between clusters to their parents in the hierarchy. The \emph{expansions stack} is used to store information on branches that have been removed and chains that have been contracted. The result of processing a cluster will be that a spanning tree of its boundary vertices (possibly including some interior vertices) is put onto the connections stack, together with the edges that connect these boundary vertices to other clusters. A cluster $Q$ is processed as follows: we pop the results from the four children of $Q$ from the connections stack; from these results we construct a graph $U$ that consists of the children's spanning trees and the edges between boundary vertices of different children; we compute a spanning tree $T$ of $U$ (using Prim's algorithm and a standard internal-memory heap as a priority queue~\cite{CLRS}); we identify dead ends and chains as described above and push them onto the expansions stack, then we remove the dead ends and contract the chains in $T$; and finally we push the results for $Q$ onto the connections stack.

Then we process the hierarchy in reverse order (top-down). We will make sure that whenever we are about to process a cluster $Q$, the top of the connections stack contains the edges of the spanning tree of $Q$ that need to be expanded---note that this is true initially because of the end result of the bottom-up phase.
To process $Q$, we pop the edges to be expanded from the connections stack, we pop the dead ends and chains in~$Q$ from the expansions stack, we expand the chains and add the dead ends to the spanning tree, and finally we divide the resulting tree among the children of $Q$: we push each child's part of the tree (including the edges from that child's cluster to neighbouring clusters) onto the connections stack. When $Q$ is a single vertex, we write its outgoing edges to the output file instead of pushing them onto the stack.


In total, we use $O(\scan(n))$ I/Os and relative I/O volume approximately~2.5.
For analysis and calculations, see the appendix.

\section{Topological sorting}\seclab{TopSort}
Given a directed acyclic grid graph $G$, we want to sort its vertices such that if there is a path from $u$ to $v$, then $u$ comes before $v$ in the sorted list. An algorithm for topological numbering (but not sorting) in $O(\scan(n))$ I/Os is known~\cite{PlanarTopSort}. It is based on first computing the lengths of the longest paths to each vertex of a suitable separator (such as our set $V_h$). Unfortunately this constitutes a significant overhead: as with our breadth-first traversal algorithm, an intermediate weighted graph $G'$ is used that is about 16 times the size of the input graph in bytes. Below we describe an approach that reduces this overhead by working with connectivity information only instead of with path lengths. We also describe how to sort (not only number) the vertices into topological order in $O(\scan(n))$ I/Os. The techniques used will prepare us for the next section on time-forward processing. For ease of description, we assume the graph is connected (details about how to handle disconnected graphs are in the appendix).

As with the algorithm in Section~\ref{sec:SimpleSSSP}, we first construct a graph $G'$ on the vertices of $G^h$ (for a well-chosen $h$). Then we compute a numbering of the vertices of $G'$; we use this numbering to cut up $G$ into chunks; and finally we sort the chunks, read the chunks in order and output all vertices. The first two steps of the algorithm are very similar to the shortest paths algorithm of \secref{SimpleSSSP}. The difference is that instead of using a weighted graph $G'$, we use an unweighted graph $G'$, and instead of computing distances from the source for every vertex of $G'$, we compute a topological numbering of the vertices of $G'$ that assigns a unique number to each vertex of $G'$. In the end we permute the topologically ordered list of vertices of $G'$ into a file $R$ indexed by $h$-number, giving a topological number $r(u)$ for each vertex $u$ of $G'$. For details, see the appendix.

To cut $G$ into chunks, we process all $h$-clusters one by one. In a cluster $Q$, each vertex $v$ of $Q$ is assigned a chunk number as follows. We will use an algorithm that assigns numbers to vertices incrementally. At any time during this algorithm, let $P(v)$ (predecessors of $v$) be the set of vertices that already have a chunk number and \emph{from} which there is a path to $v$ in $G(Q)$, and let $S(v)$ (successors of $v$) be the set of vertices that already have a chunk number and \emph{to} which there is a path from $v$ in $G(Q)$. We first give each boundary vertex $u$ chunk number $r(u)$. Then we give each vertex $v$ such that $P(v)$ is not empty, the chunk number of the highest-numbered vertex in $P(v)$. After that, we give each vertex $v$ such that $S(v)$ is not empty, the chunk number of the lowest-numbered vertex in $S(v)$. We repeat these two steps until all vertices have been numbered.
Each chunk is sorted topologically by itself, and all chunks are written to a single file $C$; we also produce a list $A$ that contains the address of each chunk in $C$ together with the topological number $r(u)$ of the chunk's boundary vertex $u$; this list $A$ is subsequently permuted so that it is ordered by topological number $r(u)$.

Finally, to output all vertices of $G$ in topological order, we simply output each chunk from $C$ in the order of their reference in $A$. The correctness of the algorithm is proven in the appendix.

In total, we use $O(\scan(n))$ I/Os and relative I/O volume less than~10.
For analysis and calculations, see the appendix.

\section{Planar time-forward processing}\seclab{TFP}
Given a directed acyclic plane grid graph $G$, we want to compute a label $\phi(v)$ of constant size for each vertex $v$; we assume we have an oracle that computes $\phi(v)$ in constant time and without I/O when the labels of the in-neighbours of $v$ are in memory. \emph{Time-forward processing} computes these labels by processing the graph in topological order, using a data structure that stores messages from one vertex to another. The idea is that whenever we compute a label $\phi(u)$, we send a message with this label from $u$ to each of its out-neighbours; thus, when a vertex $v$ is processed, the labels of its in-neighbours will be available in the data structure as messages sent to $v$. Standard solutions use priority queues to pass the messages, resulting in an I/O-complexity of $O(\sort(n))$ I/Os~\cite{PriorityQueue}.

We can obtain a solution in $O(\scan(n))$ I/Os by extending the topological-sorting algorithm from the previous section. When we make the file of chunks~$C$, we leave some extra space: for each pair of chunks such that messages must be sent from one chunk to the other, we reserve the required amount of space in the file. We annotate all edges of the graph with the address in the file where the message to be sent along that edge should be placed. Next, when we go through all chunks in topological order, we read and write messages to/from the indicated locations in the file, and produce the labels.

More precisely, for inter-cluster messages (messages passed between chunks in different clusters), we reserve space at the beginning of the file. These messages always go from one vertex of $V_h$ to an adjacent vertex of $V_h$. An address for such a message can simply be computed directly from the locations of the vertices involved, such that we reserve one address for each horizontal edge of $E_h$ (regardless of the direction), one address for each vertical edge, and one for each pair of diagonal edges (of which only one edge can be in $G$; recall that $G$ is a plane graph in this section).
For intra-cluster messages (messages passed between chunks in the same cluster), we reserve space in the file directly after the description of the receiving chunk. Addresses for the incoming intra-cluster messages are assigned in clockwise order around the boundary of the chunk. The description of the chunk lists the vertices of the chunk in topological order, with for each vertex: a vertex identifier; from which directions to receive messages; in which directions to send messages; and for all outgoing intra-cluster messages, the address where to place it, relative to the starting address of the chunk. 
To be able to output the computed labels in a structured way, we reserve space in a label file $L$ that is ordered by $h$-cluster. Within each cluster, we will store the labels chunk by chunk, for each vertex listing a vertex ID 
and its label. In the chunk file~$C$, we specify for each chunk a starting address in the label file~$L$.

Using the address file $A$ (see \secref{TopSort}), we can now traverse the chunks in topological order. We read each chunk with its incoming intra-cluster messages into memory, read the incoming inter-cluster messages from the beginning of the file, apply time-forward processing to the chunk in memory, write the results (vertex IDs and labels) to $L$, and write the outgoing messages for vertices in other chunks to their respective addresses. Upon completion, we go through the label file $L$ cluster by cluster, and output the computed labels in Z-order.

The complete algorithm uses $O(\scan(n))$ I/Os. The key part of the analysis is that, although non-sequential I/O is needed to put messages directly in the chunk file, the cost is only one non-sequential I/O for each pair of chunks such that messages are sent between them. Because the input is a plane graph, and the chunks induce connected subgraphs of it, the number of chunk pairs between which messages are passed is linear in the number of chunks, which is $O(n/\sqrt{M})$.

Assuming 64 bits labels, the relative I/O volume is less than 70 in the worst case---this is roughly in the same order of magnitude as what priority-queue-based time-forward processing needs if a topological numbering is given.
For analysis and calculations, see the appendix.

\section{Euler tours}

Given a grid graph that forms a tree, we want to compute a so-called `Euler' tour of the tree that traverses each edge of the tree twice: once in each direction. With standard techniques this can be done in $O(\sort(n))$ I/Os~\cite{Survey}.

Our approach, in $O(\scan(n))$ I/Os, is similar to topologically sorting a directed acyclic graph. Again, we use a decomposition of the graph into $h$-clusters, but the specification of $G'$ is different. As the possible points of entry into an $h$-cluster $Q$ we do not take the vertices of $\boundary(Q)$, but the incoming edges from a vertex outside $Q$ to a vertex of $\boundary(Q)$. Thus, there are roughly $12 \cdot 2^h$ possible ways to enter $Q$. Similarly, there are roughly $12 \cdot 2^h$ ways to exit $Q$. Because only one point of exit can be connected to any point of entry, we can use a very efficient representation of $G'$: instead of storing for each pair of entry and exit points whether a connection is present, we simply store one exit point for every entry point. Chunks are constructed and ordered as with topological sorting.

The relative I/O volume for this approach is less than 15.
For analysis and calculations, see the appendix.

\section{Discussion}
We saw a number of techniques that may help solving problems on grid graphs efficiently, including hierarchies of smaller priority queues (for shortest paths); hierarchies of separators (for shortest paths and minimum-spanning trees); and using the results of pre-computations (on graphs $G'$) to cut pre-defined memory-size clusters into smaller, connected chunks that can be reordered efficiently.

\paragraph{Practical potential}
Most of our algorithms seem to have good relative I/O volumes. For computing SSSP, MST, topological sorting and Euler tours, the I/O volume estimates are such that any algorithm that would include a couple of off-the-shelve sorting passes over all vertices would be hard-pressed to beat them. However, the I/O volume calculations for most of our algorithms critically depend on our choice of the block size $B$, which effectively puts a penalty on non-sequential I/O equal to the cost of reading and processing 128 KB of data. If, on a fast machine with a modern had disk, 1 MB is a more reasonable penalty, then the relative I/O volumes go up. In fact, for fixed $n$, the I/O volume of most of our algorithms is linear in $B/\sqrt{M}$---even if this factor is, by the tall-cache assumption, $O(1)$. In contrast, the I/O volume of algorithms based on simple scans, sorts, and priority queues is linear in $1/\log\frac{M}{B}$. On the other hand, if, with other hardware, $M$ is larger or $B$ is effectively lower (such as with solid-state drives), then our new algorithms benefit from this.

Our cache-oblivious algorithm to compute minimum-spanning trees is not affected by the block size: it only accesses input and output files sequentially and operates on two stacks. This is efficient regardless of the block size.

The practical potential of our algorithms is not only determined by their I/O-efficiency, but also by the number of operations on the CPU. The cache-oblivious algorithm for minimum-spanning trees needs only $O(n)$ operations: this follows directly from the fact that it runs cache-obliviously in $O(\scan(n))$ I/Os. For SSSP and topological sorting (and hence, also for BFS, time-forward processing and Euler tours), the most CPU-intensive part of the algorithm is the computation of the graph $G'$. If the input graph is planar, we can use Klein's multiple-source shortest paths algorithm~\cite{MSSP} to compute the edges across each cluster in $O(M \log M)$ time. Thus we obtain a total running time of $O(n \log M)$. However, if the input graph is not planar, we may not be able to use Klein's algorithm and we may need as much as $\Omega(M \sqrt M)$ time per cluster. This may be prohibitively expensive. To get an efficient algorithm for non-planar graphs, it may be necessary to sacrifice I/O-efficiency for CPU-efficiency by choosing the cluster sizes much smaller than one would do if I/O-efficiency were the only concern---and the resulting implementation may be much slower than an optimal implementation for planar graphs.

\paragraph{Applications.}
Apart from the problems discussed in this paper, several other fundamental and applied graph problems could be solved with the techniques presented above. For example, when a directed grid graph forms a forest, in which each edge is directed from a child to its parent, the MST technique can be adapted to compute the number of descendants of each vertex cache-obliviously in $O(\scan(n))$ I/Os. The key adaptation is that in the bottom-up phase, when dead branches are pruned and chains are contracted, vertices get a weight that corresponds to the number of vertices of incoming dead branches and incoming chains that were removed. In the top-down phase, these weights can be used to compute subtree sizes. An immediate application is \emph{single-directional flow accumulation} on a geographic terrain, modelled as a grid in which, at each vertex, any water that flows there from above, continues its way to (at most) one neighbour in the grid. 
Our report on 
flow accumulation describes the details and experiments that show that the approach is very effective in practice~\cite{SimpleFlow}.
If water is distributed to multiple neighbours, we get a directed acyclic flow network on the grid and our time-forward processing technique in $O(\scan(n))$ I/Os could be used.

The appendix discusses some more applications.

\paragraph{Topics for further research.}
In some cases there is a striking discrepancy between labelling and ordering.
For example, by adapting the MST technique
one can compute distances along a single path through a grid graph \emph{cache-obliviously} in $O(\scan(n))$ I/Os. However, for actually
producing the vertices of the path in order, in $O(\scan(n))$ I/Os, I only know a \emph{cache-aware} solution at this time.

A depth-first traversal of an \emph{undirected} planar graph can be computed in $O(\sort(n))$ I/Os~\cite{PlanarUndirectedDFS}. Combining the ideas of Arge et al.~\cite{PlanarUndirectedDFS} with our new breadth-first traversal algorithm, we may be able to improve the bound to $O(\scan(n))$ I/Os for (plane) grid graphs. \emph{Directed} graphs may also be investigated~\cite{PlanarDirectedDFS}.

Another question that remains is whether time-forward processing can also be done in $O(\scan(n))$ I/Os on a \emph{non-planar} grid graph.

To some extent our algorithms may be adapted to planar, \emph{non-grid} graphs. The relevant property of a graph in this context is that for any $h$, a graph with $n$ vertices can be partitioned into $O(n/4^h)$ clusters (sharing vertices on their boundaries), such that each cluster has $O(2^h)$ vertices on its boundary and $O(4^h)$ vertices in total; each boundary vertex is adjacent to $O(1)$ clusters; and finally, if the boundary vertices are grouped into sets that are adjacent to the same clusters, then there are only $O(n/4^h)$ groups~\cite{PlanarSSSP}.
Most of our algorithms should be adaptable to graphs that are given partitioned in this way with $4^h \leq cM$ for small enough $c$, their vertices ordered cluster by cluster; 
our $O(\scan(n))$ I/O algorithms for SSSP and MST would need a hierarchical partitioning.
For planar graphs in which all vertices have degree at most three, the required partitioning 
can be found in $O(\sort(n))$ I/Os~\cite{PlanarSeparator}.
Planar graphs with vertices of higher degree can be represented by graphs with vertices of degree at most three with standard tricks. However, in the case of directed acyclic graphs, a challenge is posed by the fact that the representative graphs are not necessarily planar. For the purposes of topological sorting, a proper partitioning can still be found~\cite{NearPlanar}, but our time-forward processing algorithm may not be applicable.

\clearpage
\appendix

\section{Details, proofs and calculations}

\subsection{Introduction}\seclab{apx:intro}

This appendix describes details of the algorithms that could not be presented in the main text of the paper, essential ingredients for proofs of the correctness and the asymptotic I/O-complexity, and calculations of the relative I/O volume of most algorithms.

The proofs of correctness and asymptotic I/O-complexity are, of course, presented here without any reservations.

The calculations of relative I/O volume are, by their nature, somewhat subjective. These calculations depend on implementation choices in the details of the algorithms and in the way information is encoded in binary files. I have tried to make these choices so that the algorithms and the representations of the input and output files are as efficient as possible but not overly complicated. In particular, for ease of description, in this paper we define the vertices of $G^h$ such that there is a double row/column of boundary vertices along the boundary of each pair of clusters. This has the advantage that clusters are perfectly aligned with the squares that appear in the recursive definition of the Z-order. Alternatively, one may use a single row/column of boundary vertices. There are advantages and disadvantages to both approaches, and I have not fully explored the effects on I/O volume. In any case the results of the I/O volume calculations should be such that they are generally within a factor two or three from optimal; this should be enough to contribute to a judgement whether an algorithm has potential for practical applications or not.

The structure of the appendix follows the structure of the main text. \secref{PrelimApx} presents some calculations that are applicable to all algorithms; Sections~\ref{sec:SSSPApx} to \ref{sec:EulerApx} discuss the algorithms from the corresponding sections in the main text; ~\secref{DiscussionApx} mentions a few more applications for our techniques.

\subsection{Preliminaries}\seclab{PrelimApx}

An unweighted, directed or undirected grid graph can be represented by specifying for each vertex, to which of its eight neighbours in the grid it has an edge in $G$. Such a graph can therefore be specified in 8 bits per vertex, for a total input size of $n$ bytes. For a weighted, directed graph we use 64-bits numbers for the weight of each possible edge, and we get an input size of $64n$ bytes. For a weighted, undirected graph, it suffices to store each edge only at say, its left end point (or, for a vertical edge, at its top endpoint), so $32n$ bytes suffice.

To keep the calculations simple, we assume the input graph is a square of $n = 2^{40}$ vertices in $2^{20}$ rows and $2^{20}$ colums. If the input graph is a rectangle with a very small or very large width/height ratio, then this could create small deviations in the calculations due to rounding effects. Recall from \secref{Intro} that we assume $M = 2^{31}$ bytes (2 GB) (or slightly more), and $B = 2^{17}$ bytes (128\,KB).

An $h$-cluster $Q$ contains $4^h$ vertices, of which $4 \cdot 2^h - 4$ vertices lie on the boundary $\boundary(Q)$. The number of $h$-clusters in a square graph of $n = 2^{40}$ vertices is $n / 4^h$, so $G^h$ contains almost $n / 4^h \cdot 4 \cdot 2^h = 4n / 2^h$ vertices. One I/O per vertex of $G^h$ adds up to an I/O volume of $4n / 2^h \cdot B = 4n / 2^h \cdot 2^{17} = 2^{19-h} n$ bytes.

When we construct a graph $G'$ with vertex set $V_h$, that contains all edges $E_h$ and one edge $(u,v)$ for every pair $u, v$ such that $u$ and $v$ are on the boundary $\boundary(Q)$ of the same $h$-cluster $Q$, then the out-degree of each vertex $u$ in $G'$ is at most $4 \cdot 2^h - 5$ (for connections to other boundary vertices in the same cluster) plus 5 (for edges in $E_h$, there can be five if $u$ lies on a corner of its cluster), makes at most $4 \cdot 2^h$. Storing the adjacency list of any vertex in $G'$ as a bit vector with one bit per edge, a single cluster needs almost $(4 \cdot 2^h)^2 = 16 \cdot 4^h$ bits, that is, $2 \cdot 4^h$ bytes, which adds up to almost $2n$ bytes in total for $G'$.

\subsection{Single-source shortest paths}\seclab{SSSPApx}\seclab{SimpleSSSPApx}

\subsubsection{Single-source shortest paths in $O(\scan(n) + \sort(n/\sqrt{M}))$ I/Os}

\paragraph{Correctness.}
The correctness of the algorithm is easy to analyse with the same arguments as in the original algorithm by Arge et al.~\cite{PlanarSSSP}.

\paragraph{Asymptotic I/O-complexity.}
In the first phase, we read $O(n/M)$ clusters of size $O(M)$, and process each of these in memory. Each cluster has $O(\sqrt{M})$ vertices on its boundary; thus we create, for each cluster, $O(M)$ edges in $G'$. Thus, $G'$ has total size $O(n/M \cdot M) = O(n)$. We read the input file sequentially, and we write $G'$ sequentially, in $O(\scan(n))$ I/Os.

In the second phase, each cluster of $G'$ is extracted from the priority once for each of its vertices. There are $O(n/M)$ clusters, each with $O(\sqrt{M})$ vertices in $G'$; thus we have $O(n/\sqrt{M})$ extractions.

After extracting a cluster $Q$, we look up the distance estimates for the vertices of $\boundary(Q)$ to find the vertex $u$ whose distance is to be made final. Since the distance estimates file $D$ is ordered by $h$-number, that is, cluster by cluster, this can be done in $O(\scan(|\boundary(Q)|)) = O(\scan(\sqrt{M}))$ I/Os. Then we read the list of outgoing edges of $u$ from $G'$ in another $O(\scan(\sqrt{M}))$ I/Os, and update distance estimates in $D$ for vertices in at most four (but usually two) clusters in another $O(\scan(\sqrt{M}))$ I/Os\footnote{Four clusters may be accessed when $u$ lies on the corner of a cluster.} Finally, we update the keys of at most four (but usually two) clusters in the priority queue. Excluding the priority queue operations, the total cost over all extractions is $O(n/\sqrt{M} \cdot \scan(\sqrt{M})) = O(\scan(n))$ I/Os. In total, there are $O(|V_h|) = O(n/\sqrt{M})$ priority queue operations. Using an I/O-efficient priority queue, these operations can be done in $O(\sort(n/\sqrt{M}))$ I/Os.

In the third phase, we read the input and $D$ sequentially and write the output file sequentially in $O(\scan(n))$ I/Os.

In total, the complete algorithms takes only $O(\scan(n) + \sort(n/\sqrt{M}))$ I/Os.

\paragraph{Relative I/O-volume.}
The input graph has size $64n$ (see \secref{PrelimApx}). The output consists of a list of 64-bits numbers, for a total size of $8n$ bytes.

We choose $h = 12$, so that each $h$-cluster consists of $4096 \times 4096$ vertices. With 64-bits numbers for the weights of the edges, we will need $64 \cdot 2 \cdot 4^h = $ 2 GB per cluster for $G'$ (see \secref{PrelimApx}). In total, $G'$ has size $64 \cdot 2n = 128n$ bytes.

In the first phase of the algorithm, we read the input file ($64n$) and write $G'$ ($128n$).

In the second phase, we mainly access $G'$, $D$, and the priority queue. In $G'$, the adjacency list of one vertex stores roughly $4 \cdot 2^h$ 64-bits numbers, in $32 \cdot 2^h = 2^{17} = B$ bytes, exactly one block. In $D$, we store $4 \cdot 2^h$ 64-bits distance estimates per cluster, that is, one block per cluster (we can use the sign bits to mark a distance as tentative or final). In this phase of the algorithm, each adjacency list in $G'$ is read exactly once, for a total I/O volume of $128n$. Whenever we access a block of $G'$, we also read and write two blocks of $D$, for a total I/O volume of $512n$. The priority queue stores one weight (64 bits) and one number identifying an $h$-cluster (32 bits suffice) per $h$-cluster; there are $n / 4^h$ of these clusters, so $12n / 4^h = 12 \cdot 2^{40 - 2h} = 12 \cdot 2^{16}$ bytes, that is, 768 KB suffices to store the priority queue; it can therefore easily be kept in main memory.

In the third phase of the algorithm, we read the input file ($64n$) and $D$ (size $8 \cdot 4n / 2^h = n/128$, negligible) sequentially, and write the output file ($8n$) sequentially.

In total we get an I/O volume of approximately $192n + 640n + 72n = 904n$, and the total size of the input and output files is $64n + 8n = 72n$. Thus the relative I/O volume is approximately $904 / 72 < 13$.

\subsubsection{Single-source shortest paths in $O(\scan(n))$ I/Os}

\paragraph{Asymptotic I/O-complexity.}
Henzinger et al.~\cite{LinearSSSP} prove that we can charge all wasted, premature work to pairs of an $h_i$-cluster ($0 < i < k$) and a vertex on its boundary, such that each pair is charged with at most one wasted call on an $h_i$-cluster of the algorithm on page~5. For this purpose, wasted work that was attempted but did not happen because priority queues ran empty, can also be charged in this way; that is, we may assume that each call of the algorithm on an $h_i$-cluster leads to $2^{h_i - h_{i-1} - 1}$ calls on $h_{i-1}$-clusters, in recursion ultimately leading to $2^{h_i - h_0 - i}$ calls on $h_0$-clusters. Since in total, there are $O(n / 2^{h_i})$ vertices on the boundaries of $h_i$-clusters, the total number of (attempted) wasted calls on $h_0$-clusters is bounded by $\sum_{i=1}^{k-1} (n / 2^{h_i} \cdot 2^{h_i - h_0 - i}) = O(n / 2^{h_0}) = O(n/\sqrt{M})$. The total number of non-wasted calls on $h_0$-clusters is equal to the number of vertices of $G'$, which is also $O(n/2^{h_0}) = O(n/\sqrt{M})$.

As with the single-source shortest path algorithm given before, each call on an $h_0$-cluster requires only $O(\scan(\sqrt{M}))$ I/Os; times $O(n/\sqrt{M})$ calls makes $O(\scan(n))$ I/Os in total.

Each call on an $h_1$-cluster does a constant number of priority queue operations on queues of constant size (at most $4^{h_1} / 4^{h_0} = 64$) for each call on an $h_0$-cluster. With a standard (not necessarily I/O-efficient) heap implementation, this results in $O(n/\sqrt{M})$ operations and therefore $O(n/\sqrt{M}) = O(\scan(n))$ I/Os.

For $1 \leq i \leq k$, the amount of work done in a call on an $h_i$-cluster, not counting the work done in the recursive calls made from it, is dominated by $O(2^{h_i - h_{i-1}})$ priority queue operations on queues of size $4^{h_i - h_{i-1}}$. Assuming a standard (not necessarily I/O-efficient) heap implementation, these operations take $O(h_i)$ time each. Charging them to the $2^{h_i - h_0 - i}$ calls on $h_0$-clusters that are made in recursion, each call on an $h_0$-cluster is charged with an amount of work of $O(h_i 2^{h_i - h_{i-1}} / 2^{h_i - h_0 - i}) = O(h_i 2^{h_0 - h_{i-1} + i})$ in calls on $h_i$-clusters. Using $h_i / h_{i-1} = 2^{h_{i-1} - h_{i-2} - 2}$, we find that the charges induced by $h_i$-clusters (for $1 < i < k$) are only half of the charges induced by $h_{i-1}$-clusters; thus the charges from $h_1$-clusters dominate, which were analysed above.

Thus, the second phase of the algorithm (the computations on $G'$) require only $O(\scan(n))$ I/Os in total, and this brings the total for the complete algorithm down to $O(\scan(n))$ I/Os as well.

\subsection{Breadth-first traversals}\seclab{BFSApx}

\paragraph{Details of the algorithm.}
To compute the shortest paths distances in an unweighted graph, we can use the solutions for single-source shortest paths presented above. Alternatively, we use the first algorithm (the algorithm that uses only a single level of clusters) with another priority queue implementation, as follows. At any time during phase two of the single-source shortest path algorithm, let $u$ be the most recently extracted vertex, and let $d'[u]$ be $d[u]$ rounded up to the nearest multiple of $2^h$. The priority queue is stored in a set of unordered lists $L_{d[u]},...,L_{d'[u]}$ and $H_0,...,H_{2^h}$, such that $L_i$ contains all elements with key $i$, and $H_i$ contains all elements with keys more than $d'[u] + i 2^h$ and at most $d'[u] + (i+1) 2^h$. In principle, the priority queue does not support updates to any keys stored in it: when the key of an $h$-cluster $Q$ needs to be updated, we simply reinsert $Q$ and leave the old copy in the queue. The old copy is simply discarded when it is extracted and it turns out that $\boundary(Q)$ does not have a vertex with tentative distance equal to the key of $Q$. Whenever $d'[u]$ increases, the lists $L$ and $H$ are updated by distributing the first non-empty list $H_0$ over the lists $L$, and renumbering the remaining lists $H$.

In practice, there is a little catch. With a priority queue implementation that supports updates to the keys stored in it, we could make sure that each $h$-cluster is stored in the queue only once, so that the size of the queue is limited to $O(n/M)$. In practice this means the queue fits in memory and does not cause any I/O at all. With the alternative implementation given above, the queue may grow to size $O(n/\sqrt{M})$ (below we will see that the number of operations on the queue is bounded by $O(n/\sqrt{M})$). For very large inputs, this may not fit in memory. Although the I/O-complexity of the whole algorithm will be guaranteed to be $O(\scan(n))$ nonetheless, it will harm its performance. We may remedy this by modifying the alternative priority queue as follows: when we want to update the key of a cluster $Q$, then, if the part of the priority queue that contains the reference to $Q$ is currently in memory, we update it, otherwise we leave the old entry where it is and reinsert~$Q$ with the new key.

\paragraph{Correctness.}
In the current application of the single-source shortest paths computation, the weights of all edges of $G$ are 1, and therefore the weights of all edges of $G'$ are natural numbers and at most $4^h-1$, the size of an $h$-cluster minus 1. Thus, at any time during the algorithm, the keys of the $h$-clusters in the priority queue are natural numbers that differ by at most $4^h - 1$. Therefore the lists $L$ and $H$ of the alternative priority queue implementation suffice to store the elements in the queue.

\paragraph{Asymptotic I/O-complexity.}
According to the analysis in \secref{SimpleSSSPApx}, computing single-source shortest paths takes $O(\scan(n) + \sort(n/B))$ I/Os. The second term was due to the priority queue operations. However, with the new priority queue implementation, these operations will be more efficient. The number of lists $L$ and $H$ is only $O(\sqrt{M}) = O(M/B)$, so that we can keep at least one block of each of them in memory, allowing insertions and extractions in $O(1/B)$ I/Os on average. Each element in the queue is moved from a list $H$ to a list $L$ at most once, at a cost of $O(1/B)$ I/Os per element moved. The number of insertions is still at most four times the number of non-discarded extractions, and the number of discarded extractions is at most the number of insertions; thus the total number of priority queue operations is $O(n/\sqrt{M})$. Thus the priority queue operations take only $O(\scan(n/\sqrt{M}))$ I/Os in total, and the complete single-source shortest paths algorithm runs in $O(\scan(n))$ I/Os.

To put the vertices of the graph in breadth-first traversal order, we first cut up the clusters into chunks. $\Theta(n/M)$ $h$-clusters are read in $O(M/B)$ I/Os each, and for each $h$-cluster, $O(\sqrt{M})$ chunks of size $O(\sqrt{M})$ are written to disk; this takes $O(\scan(n))$ I/Os and results in $O(n/\sqrt{M})$ chunks in total.

If the chunk address list $A$ fits in memory, no I/O is needed to sort it. Otherwise we have $n > M$, and thus $(n/M^{3/4})^2 = n^2 / M^{3/2} > n / \sqrt{M}$. Therefore, sorting the chunk addresses with radix sort requires only two passes, each taking $O(n / \sqrt{M}) = O(\scan(n))$ operations.

Finally, reading the chunks takes $O(\scan(n) + n/\sqrt{M}) = O(\scan(n))$ I/Os. Because the distances from $s$ of the vertices in a single chunk differ by at most $O(\sqrt{M})$, only $O(\sqrt{M})$ stacks are non-empty at any time during the algorithm. Therefore we can always keep one block of each non-empty stack in memory, and do all $\Theta(n)$ stack operations in $O(1/B)$ I/Os per operation on average, for a total of $O(\scan(n))$.

In total, the whole algorithm takes $O(\scan(n))$ I/Os.

\paragraph{Relative I/O-volume.}
The input graph has size $n$ (see \secref{PrelimApx}). The output consists of a list of 64-bits vertex IDs, for a total size of $8n$ bytes.

As with the single-source shortest path algorithm (see \secref{SSSPApx}), we choose $h = 12$. The input file now has size $n$ instead of $64n$, and because 32 bits numbers suffice for the weights of $G'$ (edge weights are not more than roughly $4^h$), the total size of $G'$ will be roughly $64n$ instead of $128n$. Thus the shortest path computations now take an I/O volume of only $65n + 576n + n = 642n$, excluding writing the output.

Next we cut up clusters into chunks. The input, by cluster, can be piped through directly from the shortest path computations. The output consists of a sequence of chunks. Each chunk has one starting vertex for which we need to store an identifier and its distance from $s$. The other vertices (the vast majority) can be given by storing, for each vertex, which of its neighbours in the chunk are its children in the `local' breadth-first search tree: this can be done with 8 bits = 1 byte per vertex. Thus the file $C$ has size roughly $n$. The total number of chunks is less than $5 n/2^h$, that is, less than $n/819$; thus the I/O that is needed to sort the chunk address list $A$ with I/O-efficient merge sort is negligible.

In the final part of the algorithm, reading $A$ causes negligible I/O, reading $C$ results in an I/O volume of $n$ for sequential I/O and $5n/2^h \cdot B = 160n$ for random access to the starting nodes of the chunks, each 64-bits vertex identifier may be subject to two stack operations (I/O volume $16n$) and is finally written to the output (I/O volume $8n$).

In total we get an I/O volume of approximately $642n + n + 185n = 828n$, and the total size of the input and output files is $n + 8n = 9n$. Thus the relative I/O volume is approximately $828 / 9 < 100$.

\subsection{Minimum-spanning trees}\seclab{MSTApx}

The correctness of the minimum-spanning tree algorithms is based on the following lemma:

\begin{lemma}\label{lem:subtree}
For a grid graph $G$ and any $h$, let $G_1,...,G_k$ be the graphs induced by the $h$-clusters $Q_1,...,Q_k$ of $G$, and let $T_1,...,T_k$ be minimum spanning trees of $G_1,...,G_k$, respectively. Then the union $U$ of $T_1,...,T_k$ and $G^h$ contains a minimum spanning tree of $G$.
\end{lemma}
\begin{proof}
Let $T$ be a minimum spanning tree of $G$. Suppose $T$ contains an edge $(u,v)$ of $G_i$ that is not included in $T_i$. Observe that $(u,v)$ must be at least as heavy as any edge of the path from $u$ to $v$ in $T_i$, otherwise we could get a lighter spanning tree of $G_i$ by inserting $(u,v)$ into $T_i$ and removing a heavier edge of the resulting cycle. Now we can modify $T$ as follows: remove $(u,v)$, separating $T$ into a component $T_u$ containing $u$ and a component $T_v$ containing $v$. Now the path from $u$ to $v$ in $T_i$ must contain at least one edge with one endpoint in $T_u$ and one endpoint in $T_v$: use this edge, which is at most as heavy as $(u,v)$, to reconnect $T_u$ and $T_v$. Thus we obtain a spanning tree of $G$ that is at most as heavy as $T$. In this way we can replace every edge of $T$ that is not included in $U$ by an edge that is included in $U$ and is not heavier, and we obtain a minimum spanning tree of $G$ that is included in $U$.
\end{proof}

\subsubsection{Minimum-spanning trees cache-aware in $O(\scan(n) + \sort(n/\sqrt{M}))$ I/Os}

\paragraph{Asymptotic I/O-complexity.}
Choosing $h$ as large as possible such that a minimum-spanning tree computation on $G_h$ can still be done in main memory, the algorithm described in \secref{SimpleMST} can be implemented to run in $O(\scan(n) + \sort(n/\sqrt{M}))$ I/Os: first process all $O(n/M)$ $h$-clusters $G_i$ one by one in memory to compute the trees $T'_i$, which each have size $O(\sqrt{M})$; thus $U'$ has size $O(n/\sqrt{M})$ and a minimum spanning tree of $U'$ is computed in $O(n/\sqrt{M} + \sort(n/\sqrt{M}))$ I/Os (using Prim's algorithm with an I/O-efficient priority queue); finally $T$ is constructed by expanding chains and adding dead ends cluster by cluster in $O(\scan(n))$ I/Os.

\paragraph{Relative I/O-volume.}

As explained in \secref{PrelimApx}, the input file has size $32n$. We will produce the output in the same format.

We choose $h = 12$, so that each $h$-cluster consists of $4096 \times 4096$ vertices, requiring $32 \cdot 4^h =$ 512 MB of storage.

We will produce $U'$ in the form of a list of vertices, which stores for each vertex its incident edges (other endpoints and weights), consecutively; in this list we reserve space for $8 \cdot 2^h$ vertices per cluster ($4 \cdot 2^h$ vertices on the boundary and $4 \cdot 2^h$ in the interior), each with at most 8 incident edges, each requiring 16 bytes; thus $U'$ contains $8n/2^h = n/512$ vertices in total, stored in $n/512 \cdot 8 \cdot 16 = n/4$ bytes.
The number of undirected edges in each of the spanning trees $T_1,...,T_k$ is at most $8 \cdot 2^h$; that makes $16 \cdot 2^h$ when counting them in both directions.

In the first phase (computing $U'$), we read the input file ($32n$) and write $U'$ (negligible).

The second phase (computing $T'$ from $U'$) requires accessing each vertex once over every incident edge, and once to add it to the tree. Assuming that the vertices on the boundary of each cluster are stored in order of $h$-number, the (almost always) at most three edges\footnote{on a corner there may be five edges} from a boundary vertex into a neighbouring cluster are (almost always) stored in one block; other edges may lead to vertices in different blocks. Thus we have $8n/2^h$ I/Os to access vertices of $U'$, $16n/2^h$ I/Os to follow edges of $T_1,...,T_k$, and $4n/2^h$ I/Os to follow edges of $G^h$; this makes $28n/2^h$ I/Os for a volume of $28n/2^h \cdot B = 896n$. The priority queue may not fit in memory, but it is still so small that the I/O needed for it is negligible.

In the last phase of the algorithm (computing $T$) we read the input file again, and $T'$, and write the output file.

In total we get an I/O volume of approximately $32n + 896n + 64n = 992n$, and the total size of the input and output files is $32n + 32n = 64n$. Thus the relative I/O volume is approximately $992 / 64 < 16$.

With an input file of up to 128 GB ($n = 2^{32}$), the graph $U'$ and the priority queue would fit in memory and the second phase does not need non-sequential I/O at all. The total relative I/O volume would then be roughly $(32+32+32)/(32+32) = 1.5$.

\subsubsection{Minimum-spanning trees cache-obliviously in $O(\scan(n))$ I/Os}

\paragraph{Asymptotic I/O-complexity.}
When processing an $h$-cluster $Q$ in the bottom-up phase, we read $O(2^h)$ vertices and edges from the connections stack, and push $O(2^h)$ vertices and edges onto the connections and the expansions stack; in the top-down phase we read $O(2^h)$ vertices and edges from the connections and expansions stacks, and push $O(2^h)$ vertices and edges onto the connections stack. An $h$-cluster contains $4^{h-i}$ $i$-clusters. The total number of stack operations in any $h$-cluster and its children is therefore $\sum_{i=0}^h O(4^{h-i} \cdot 2^i) = O(4^h)$.

Let $h = \frac12\log M - O(1)$ be the largest $h$ such that $4^h \leq cM$ (for a constant $c$ that is small enough). The number of I/Os needed to process an $h$-cluster $Q$ and all its descendants in the hierarchy is $O(\scan(M))$, since it reads and writes a contiguous part of size $O(M)$ of the input and output file, and it affects at most the top parts of size $\Theta(M)$ of the stacks, which could be buffered in main memory while processing $Q$. Thus, all $h$-clusters together are processed in $O(n/M \cdot \scan(M)) = O(\scan(n))$ I/Os in total.

When processing a larger cluster, that is, an $i$-cluster with $i > h$, we do not access the input and output files, but only work on the stacks. In the bottom-up phase, we read $O(2^i)$ vertices and edges from the connections stack, then we compute a minimum spanning tree of these edges, using Prim's algorithm with a heap of size $O(2^i)$, and then we prune and contract the tree and write $O(2^i)$ vertices and edges to the connections and expansions stacks. Reading, pruning, contracting and writing can all be done in linear time, and thus, in at most $O(2^i)$ I/Os. The computation of the minimum spanning tree takes $O(2^i \log 2^i) = O(i \cdot 2^i)$ time. However, because we can keep the top part of height $\log M - O(1)$ of the heap in memory during this computation, we need only $O(\max(i - \log M, 1) \cdot 2^i)$ I/Os. Processing the $i$-cluster in the top-down phase takes only linear time, and thus, only $O(2^i)$ I/Os. Adding up over all $i$-clusters with $i > h$, we get an I/O-bound of $\sum_{i = h+1}^{...} O(n/4^i \cdot \max(i - \log M, 1) \cdot 2^i) = \sum_{i = h+1}^{...} O(n/2^i \cdot \max(i - \log M, 1))$. Since $2^i$ grows at least a constant factor faster than $\max(i - \log M, 1)$, the sum is dominated by the first term, and the total number of I/Os during the processing of larger clusters is $O(n/\sqrt{M}) = O(\scan(n))$.

\paragraph{Relative I/O-volume.}
Assuming the grid graph is a square of $2^{20} \times 2^{20}$ vertices, processing the single 20-cluster, which contains the full grid, requires finding a minimum spanning tree on a graph of at most $4 \cdot 8 \cdot 2^{19} = 2^{24}$ vertices and less than $2^{25}$ edges. With a normal, 64-bits-pointer-based adjacency-list data structure and a heap-based priority queue (cross-linked with the vertex list), we need 32 bytes per vertex and 24 bytes per edge. So we need less than 1.5 GB. This still fits in main memory; thus the only disk access needed by the cache-oblivious algorithm is to the input and output files (which are read and written once, in order) and the connections and expansions stacks.

On any level below the top, we can process a cluster with only 1 GB of main memory. This leaves almost 1 GB which we can use to keep the connections stack in memory. The maximum size of the part of the connections stack that is used while processing an $i$-cluster in the bottom-up phase is $3 \cdot \sum_{j=0}^{i-1} T(j)$, where $T(j)$ is the size of the data pushed on the stack by a $j$-cluster. In the worst case, $T(j)$ must accommodate $8 \cdot 2^j$ edges for the spanning tree on $4 \cdot 2^j$ boundary vertices (using 24 bytes for each edge to identify endpoints and weight) and weights for $12 \cdot 2^j$ edges from boundary vertices to other clusters (8 bytes per edge): in total $288 \cdot 2^j$ bytes. Thus $3 \cdot \sum_{j=0}^{i-1} T(j) < 864 \cdot 2^i$, which is less than 1 GB all the way up to $i = 20$. Therefore the connections stack can be kept in memory. The top-down phase is symmetric.

Thus, practically the only I/O that needs to be done on the stacks is to write and read data on the expansions stack once. Each vertex can become a dead end once (so that one edge needs to be written to the expansions stack) and each vertex can be removed by contraction once (so that two edges may be written to the expansions stack); thus the number of edges (for 24 bytes: two vertex IDs and weight) that can be written to the expansions stack is at most twice the number of vertices. Therefore the total I/O for writing and reading the expansions stack is at most $96n$.

In total we get an I/O volume of approximately $32n + 96n + 32n = 160n$, and the total size of the input and output files is $32n + 32n = 64n$. Thus the relative I/O volume is approximately $160 / 64 = 2.5$.


If the input would be bigger ($n = \Omega(M^2)$), then $O(n/M)$ non-sequential I/Os may become necessary when processing clusters at the top of the hierarchy. This is negligible.

\subsection{Topological sorting}\seclab{TopSortApx}

\paragraph{Details of the algorithm.}
To construct the graph $G'$, we process the $h$-clusters one by one. For each pair of vertices $u,v$ on the boundary of an $h$-cluster $Q$, we include an edge $(u,v)$ in $G'$ if and only if there is a path from $u$ to $v$ in $G(Q)$. Furthermore, a vertex $u$ on the boundary of $Q$ may contain edges in $G$ to a constant number of neighbours in other clusters; these are also included in $G'$. Which edges are included in $G'$ is specified by storing a string of bits for $u$, which indicates, for each possible destination $v$, whether an edge $(u,v)$ exists. These bit strings are ordered by $h$-number. In addition, we create a file $D$, ordered by $h$-number, that stores the number of incoming edges for every vertex of $G'$. While doing so, the $h$-numbers of vertices of $G'$ with in-degree zero are inserted in a queue $Z$.

After constructing $G'$, we compute a topological numbering of its vertices, as follows. As long as $Z$ is not empty, we extract a vertex $u$ from $Z$, write $u$ to an output file $T'$, and then we look up the out-neighbours of $u$ in $G'$, decrease their in-degrees in $D$ by one, and add any vertices of $G'$ that have their in-degree lowered to zero to $Z$. After $Z$ has run empty, the positions of the vertices in $T'$ constitute the numbering we are after. We permute $T'$ into a file $R$ indexed by $h$-number, giving a topological number $r(u)$ (the position in $T'$) for each vertex $u$ of $G'$.

The rest of the algorithm is explained in \secref{TopSort}. When the graph is disconnected, the given algorithm to cut a cluster into chunks may fail to number all vertices. In that case we proceed as follows. In a given $h$-cluster $Q$, let $L$ be the `left-over set', that is, the set of unnumbered vertices that remain. The graph $G(L)$ induced by $L$ may have a number of weakly-connected components; we assign its vertices to chunks such that each weakly-connected component $K$ of $G(L)$ is in the same chunk as a vertex $u$ such that $v \in K$, and $u$ is the left neighbour of $v$ in the grid (even if there is no edge between $u$ and $v$).

\paragraph{Correctness.}
The correctness of the algorithm can be shown through the following lemma:
\begin{lemma}
The assignment of chunk numbers to vertices of $G$ is such that the vertices of any path in $G$ are in order of non-decreasing chunk number.
\end{lemma}
\begin{proof}
We need to show that if $(u,v)$ is an edge in $G$, then the chunk number of $v$ is at least as high as the chunk number of $u$. If $u$ and $v$ are in different $h$-clusters, then $u$ and $v$ are vertices of $G^h$ and $(u,v)$ is an edge in $G^h$, so the topological sorting of $G'$ must have assigned a higher number to $v$ than to $u$. Otherwise, $u$ and $v$ lie in the same $h$-cluster $Q$.

If $u$ and $v$ are boundary vertices of $Q$, then, again, $u$ and $v$ are vertices of $G^h$ and $(u,v)$ is an edge in $G^h$, so the topological sorting of $G'$ must have assigned a higher number to $v$ than to $u$.

If $u$ is on the boundary of $Q$ but $v$ is not, then in the chunk assignment phase, $v$ will immediately get a chunk number at least as high as that of $u$, because $u \in P(v)$.

If $v$ is on the boundary of $Q$ but $u$ is not, then $u$ can get its number in two ways. The first possibility is that immediately after assigning chunk numbers to the boundary vertices of $Q$, the set $P(u)$ is not empty and $u$ gets the chunk number of a boundary vertex $t$ such that there is a path from $t$ to $u$ in $G(Q)$. This implies that there is also a path from $t$ to $v$. Thus $G'$ contains an edge $(t,v)$, and the topological sorting of $G'$ must have assigned a higher number to $v$ than to $t$ (and consequently, $u$). The second possibility is that $P(u)$ is initially empty. But $S(u)$ contains at least $v$, so $u$ will be assigned a chunk number at least as low as that of $v$.

If neither $u$ nor $v$ is on the boundary of $Q$, we distinguish five cases.\begin{compactenum}
\item $v$ did not get its number before $u$, and $u$ got its chunk number as the highest number of any vertex in $P(u)$. Then, at the same time, $v$ must have gotten a chunk number that is at least as high, since $P(v) \supset P(u)$.
\item $v$ did not get its number before $u$, and $u$ got its chunk number as the lowest number of any vertex in $S(u)$. Then, if $v$ got its number in the same round as $u$, then $v$ got a number that is at least as high, since $S(v) \subset S(u)$; otherwise, $v$ must have gotten its number in the next round, since $P(v)$ included $u$ at that time, and the number of $v$ will be at least the number of $u$.
\item $u$ did not get its number before $v$, and $v$ got its chunk number as the lowest number of any vertex in $S(v)$. This case is symmetric to the first case.
\item $u$ did not get its number before $v$, and $v$ got its chunk number as the highest number of any vertex in $P(v)$. This case is symmetric to the second case.
\item If the above cases do not apply, then $u$ and/or $v$ must have gotten their numbers in the final step, when vertices of the left-over set $L$ get their numbers. Since there is an edge from $u$ to $v$, the vertices $u$ and $v$ must be in the same weakly-connected component of $G(L)$, and thus they got the same chunk number.
\end{compactenum}
\end{proof}
Since the algorithm outputs the vertices of $G$ chunk by chunk in order by chunk number, and within each chunk, in topological order, the above lemma implies directly that the algorithm outputs the vertices of $G$ in topological order.

\paragraph{Asymptotic I/O-complexity.}
The analysis of the construction of $G'$ is similar to the analysis of the first phase of the single-source shortest path algorithms. The main difference is that, to obtain the in-degrees of the vertices on the boundary of a cluster, we also need to access vertices on the boundaries of adjacent clusters. As a very conservative estimate, this may have the effect that blocks that contain boundary vertices of clusters are accessed eight more times when processing neighbouring clusters. This amounts to an additional $O(\scan(n))$ I/Os.

In the next step, topologically sorting $G'$, we extract $|V_s| = O(n/\sqrt{M})$ vertices from $Z$; operations on $Z$ and $T'$ take $O(1/B)$ I/Os each on average; the dominant terms in this step are accesses to $G'$ and $D$ and permuting $T'$ into $R$. Because $D$ is ordered by $h$-number, accesses to $G'$ and $D$ take $(\scan(\sqrt{M}))$ I/Os for every vertex extracted from $Z$, which makes $O(\scan(n))$ in total. Permuting $T'$ into $R$ costs $O(\min(|T'|, \sort(|T'|))$ I/Os; since $T'$ contains $O(n/\sqrt{M})$ vertices, this is at most $O(\scan(n))$).

To make the chunks we read the input graph and $R$ sequentially and we write $C$ and $A$ sequentially, all in $O(\scan(n))$ I/Os.

Afterwards, we permute $A$ in, again, $O(\scan(n))$ I/Os.

Finally, to produce the output, we read $O(n/\sqrt{M})$ chunks of total size $O(n)$ in $O(n/\sqrt{M} + \scan(n)) = O(\scan(n))$ I/Os.

The whole algorithm runs in $O(\scan(n))$ I/Os.

\paragraph{Relative I/O-volume.}
The input graph has size $n$ (see \secref{PrelimApx}). The output consists of a list of 64-bits vertex IDs, for a total size of $8n$ bytes.

When processing an $h$-cluster in memory, we will need $4^h$ bytes for the input, and $2 \cdot 4^h$ bytes for the edges from $G'$ in this cluster (see \secref{PrelimApx}). We choose $h = 14$, so that each $h$-cluster requires 768 MB of working space.

The number of incoming edges of a vertex in $G'$ can be stored as a 16-bit number, so that the size of $D$ is 2 bytes per vertex of $G'$, which is $2 \cdot 4n / 2^h$ bytes $=$ 512 MB. We will keep $D$ in memory throughout the algorithm. We will store $G'$ by giving for each vertex, in order of their $h$-numbers, a bit vector that specifies its outgoing edges (see \secref{PrelimApx}).

In the first phase (building $G'$) we read each $h$-cluster from the input file into memory once (sequentially), and write the bit vectors in $G'$ (sequentially), the incoming edge counts in $D$ (in memory); and one or more vertices to the queue $Z$ (sequentially). To be able to write the incoming edge counts, we also need to read the blocks that contain the boundary vertices from other clusters that have outgoing edges to the current cluster. Each block of size $2^{17}$ contains two 8-clusters (recall that the input is in Z-order), so each $h$-cluster is stored in $4^h / 2^{17} = 2048$ blocks, and the number of blocks adjacent to a cluster is at most four in the corners plus $4 \cdot 2^{14} / 2^8 = 256$ along the edges. The extra I/O incurred by accessing boundary vertices of neighbouring $h$-clusters will therefore be negligible in the end result of this calculation. The sequential access to $Z$, to which eventually only $|V_h| = 4n/2^h$ vertices will be written, incurs negligible I/O. The total I/O-volume per $h$-cluster in this phase is therefore approximately $n$ (for reading the input once) plus $2n$ (for writing $G'$), which makes $3n$ in total.

In the next phase (topologically sorting $G'$), reading and writing $Z$ and $T'$ causes negligible I/O, and all access to $D$ is for free because it is kept in memory. Each vertex $v$ extracted from $Z$ results in one random access to $G'$ to retrieve its outgoing edges. Thus we get an I/O volume of $2^{19-h} n = 32n$ (see \secref{PrelimApx}). Permuting $T'$ into $R$ using I/O-efficient merge sort causes negligible I/O.

In the third phase (making chunks), reading $R$ and writing the chunk addresses to $A$ is negligible. The input file (of size $n$) is read once, and all vertex IDs are written to a chunk file of size $4n$ (32 bits numbers for vertex IDs suffice, if we use local IDs, relative to the cluster). Permuting the chunk addresses in $A$ takes negligible I/O when done with I/O-efficient merge sort.

In the final phase of the algorithm (putting everything in order) we use one random accesses per vertex of $V_h$ (to start reading its chunk); larger chunks may also require sequential access to the chunk file. Thus we get an I/O-volume of $2^{19-h} = 32n$ for random access and $4n$ for sequential access to the chunk file, plus $8n$ to write the result.

In total we get an I/O volume of approximately $3n + 32n + 5n + 44n = 84n$, and the total size of the input and output files is $n + 8n = 9n$. Thus the relative I/O volume is approximately $84 / 9 < 10$.

\subsection{Planar time-forward processing}\seclab{TFPApx}

\paragraph{Asymptotic I/O-complexity.}
The analysis of the I/O-complexity mostly follows that of the topological-sorting algorithm from the previous section. The only differences are when we process the chunks in topological order, and when we sort the label file.

First, we need to retrieve inter-cluster messages from the beginning of the file for each chunk (this takes O(1) I/Os per chunk, which makes $O(n/\sqrt{M})$ in total).

Second, we need to output the results: in total we write $O(n)$ labels, to $O(n/\sqrt{M})$ different contiguous sections of $L$, in $O(\scan(n))$ I/Os.

Third, we need to write outgoing messages. In total, $O(n)$ messages may need to be written. Because the input, together with the neighbourship relations that are used to number left-over chunks, constitutes a plane graph, and the chunks induce disjoint connected subgraphs of this graph, the number of chunk pairs between which messages must be passed is linear in the number of chunks, which is $O(n/\sqrt{M})$.
For each such adjacency relation, messages need to be written once, namely when the first of the two chunks is processed, and the messages are written to a contiguous set of addresses. Therefore, writing all messages takes $O(n/\sqrt{M} + \scan(n)) = O(\scan(n))$ I/Os in total.

Fourth, we need to put the labels into Z-order. Since $L$ is ordered by cluster, we can simply read $L$ sequentially cluster by cluster, reordering each cluster in memory, and writing the reordered clusters to the output file sequentially, in $O(\scan(n))$ I/Os.

Thus the complete algorithm takes $O(\scan(n))$ I/Os in total.

\paragraph{Relative I/O-volume.}
The input graph has size $n$ (see \secref{PrelimApx}). We assume the labels to be computed are 64 bits numbers, so the output file has size $8n$.

We choose $h = 13$, so that each $h$-cluster consists of $8\,192 \times 8\,192$ vertices, and we have enough space (32 bytes per vertex) to compute chunks and message addresses when processing a cluster in memory. Note that messages addresses are all relative within a cluster, which has only $4^h = 2^{26}$ vertices, so that 32-bits numbers are sufficient.

The analysis of the first phase (constructing $G'$) is as in \secref{TopSortApx} (I/O volume $3n$).

The number of accesses to $G'$ in the second phase (computing a topological numbering of $G'$) is twice as large as in \secref{TopSortApx} (because the clusters are smaller), resulting in an I/O volume of $64n$.

The I/O volume of the third phase (computing chunks and message addresses) is dominated by reading the input file (size $n$) and writing the chunk file $C$. Since the input graph is planar, the average number of outgoing edges per vertex is at most three, and the average number of incoming edges per vertex is at most three. Therefore, per vertex in the chunk file, on average we may need at most roughly 42 bytes: 4 bytes for a vertex ID (relative to the chunk's boundary vertex), 2 bytes for the bit vectors specifying the incoming and outgoing edges, 12 bytes for the 32 bits message addresses of three outgoing edges, and 24 bytes for the 64 bits labels of the in-neighbours. Thus we have an I/O volume of at most $n + 42n = 43n$ in the third phase.

In the next phase (the actual time-forward processing), the dominant terms are $42n$ for reading the cluster file (and incoming messages), $12n$ for writing the computed labels to $L$ (8 bytes per label, and 4 bytes for a vertex ID, local within an $h$-cluster), $24n$ for writing outgoing messages, and $384n$ for, on average, at most six random accesses per chunk (one for reading the chunk, one for reading the incoming inter-cluster messages, one for writing the computed labels to $L$, and three for writing outgoing intra-cluster messages (recall from \secref{PrelimApx} that one random access per vertex of $G'$ contributes $2^{19-h}n = 64n$ to the I/O volume). Thus we have an I/O volume of $462n$ in this phase.

Finally, we read $L$ ($12n$) and write the final output for an I/O volume of $12n + 8n = 20n$.

In total we get an I/O volume of approximately $3n + 64n + 43n + 462n + 20n = 592n$, and the total size of the input and output files is $n + 8n = 9n$. Thus the relative I/O volume is approximately $592 / 9 < 70$.

\paragraph{Comparison to priority-queue based time-forward processing.}
It may be interesting to try to compare the I/O volume of our new time-forward processing algorithm for grid graphs to traditional time-forward processing with a priority queue.

To make it easy, let us assume the input graph is given together with topological numbers for each vertex ($n$ bytes for edges, $8n$ bytes for topological numbers). By a three pass $M/B$-way merge sorting step we can turn it into an input graph in topological order, given in $40n$ bytes (for each vertex a 64 bits ID, a 64 bits topological number, and on average, three 64 bits topological numbers for its out-neighbours). The sorting requires an I/O volume of roughly $9n + 40n$ for the first pass and $40n + 40n$ for each the remaining passes, for a total of $209n$.

Reading the resulting, sorted graph and writing pairs of vertex IDs plus labels to an output file results in an I/O volume of $56n$. The priority queue may need to be three levels deep. Assuming that on average, in the priority queue each of $3n$ messages is written to disk and read into memory three times, this creates an I/O volume of $16 \cdot 6 \cdot 3n = 288n$. The total I/O volume for the actual time-forward processing phase is thus $288n + 56n = 344n$.

Subsequently sorting the output by $M/B$-way merge sort in three passes takes $6 \cdot 16n - 8n = 88n$ (in the last pass the vertex IDs can be omitted).

In total we get an I/O volume of approximately $209n + 344n + 88n = 641n$; thus the relative I/O volume is approximately $641 / 9 \approx 70$.

\medskip
Therefore it is hard to make a comparison. The calculation for our new algorithm assumes the worst case for the number of intra-cluster messages and the random accesses needed to write them. In practice, the number of intra-cluster messages may easily be much smaller and many may be written between chunks that are neighbours in the chunk file, reducing random I/O. Also, chunks that are consecutive in the topological order may often be close to each other in the chunk file and the label file too, further reducing random I/O.

The calculation for the traditional time-forward processing algorithm makes a wild guess at the amount I/O required for priority queue operations. In practice, this may well be much better (if the priority queue is never big) or much worse. Furthermore, the calculation omits the effort required to topologically sort the graph if no topological numbering is given.



\subsection{Euler tours}\seclab{EulerApx}

\paragraph{Relative I/O-volume.}
The input graph has size $n$ (see \secref{PrelimApx}). Because each vertex in the tour is a neighbour of the previous one, each vertex except the first one in a chunk, and in the final output, can be specified by a 8-bit number. The tour contains $2n$ edges; thus the output size is $2n$.

In comparison to the general topological sorting algorithm analyse in \secref{TopSortApx}, we now use a more efficient representation of $G'$: with one 32 bits number per point of entry, we only need $48 \cdot 2^h$ bytes per cluster, and $48n / 2^h$ bytes in total. The size of $G'$ is therefore negligible.

We can also fit larger clusters in main memory; we can now choose $h = 15$. This reduces the cost of random accesses while computing a topological numbering of $G'$ and while putting the vertices in order: from $32n$ to $16n$ in each of these two steps.

With the calculation from \secref{TopSortApx} changed accordingly, we find that we can compute an Euler tour of a tree in a grid graph in an I/O volume of roughly $n + 16n + 3n + 20n = 40n$, and a relative I/O volume of $40n / (n + 2n) < 15$.

Note that if we would produce full vertex IDs in the final output, the ratio would improve to $46/9 < 6$.

\subsection{Further applications}\seclab{DiscussionApx}
As mentioned in the introduction, algorithms to label the connected components in a grid graph in $O(\scan(n))$ I/Os were already known. With our techniques one can also \emph{sort} a grid graph into its connected components. For this purpose, one can adapt the topological-sorting algorithm. Instead of processing $G'$ to obtain a topological numbering of the vertices of $G^h$, one would process $G'$ to label the vertices of $G^h$ such that vertices in the same component get the same number. Each cluster $Q$ can be cut up into chunks that lie in the same component of $G$, plus an additional chunk with components that lie inside $Q$ and are not connected to any vertices outside~$Q$. These chunks can then be put in order in the same way as the chunks in the topological-sorting algorithm, in $O(\scan(n))$ I/Os.

Our minimum-spanning tree algorithms can be adapted to another application in hydrological analysis. Under the assumption that water always flows downhill and does not disappear in any other way, depressions in a terrain fill up with lakes. For a terrain that is given as a grid of cells of which the elevation is known, we could compute how far the water in each cell would rise before the lake overflows and the water finds a way out. The water levels of the lakes can in fact be computed through a minimum-spanning tree computation: make an edge between every pair of neighbouring cells, and give each edge an elevation equal to the elevation of the highest of the two cells. This creates a graph of which a minimum-spanning tree $T$ can be computed. One can now prove that in each cell $c$, the water will rise to the level of the highest edge on the path in $T$ from $c$ to the sea. The minimum-spanning tree algorithm can easily be adapted to compute the water levels of all cells on the fly while computing the minimum-spanning tree. Our report on flow accumulation gives more details on how to implement this~\cite{SimpleFlow}.

\end{document}